\documentclass[12pt, draftclsnofoot, onecolumn]{IEEEtran}
\usepackage{pgfplots}
\pgfplotsset{compat=1.7}
\usepackage{subcaption}
\usepackage[margin=1in,letterpaper]{geometry}
\usepackage[T1]{fontenc}
\usepackage{lmodern}
\usepackage{amsmath}
\usepackage{amssymb}
\usepackage{amsthm}
\usepackage{color,soul}
\usepackage{xcolor}
\usepackage{graphicx}
\usepackage{algorithm}
\usepackage{tcolorbox}
\usepackage{graphicx}
\usepackage{subcaption}
\usepackage{pgfplots}
\pgfplotsset{%
    ,compat=1.12
    ,every axis x label/.style={at={(current axis.right of origin)},anchor=north west}
    ,every axis y label/.style={at={(current axis.above origin)},anchor=north east}
    }

\usepackage{tkz-graph}
\usepackage{tikz}
\usepackage{algpseudocode}
\usepackage[colorlinks=true, linkcolor=red, urlcolor=blue, citecolor=gray]{hyperref}

\usepackage{comment}
\newtheorem{thm}{Theorem}
\newtheorem{cor}{Corollary}
\newtheorem{rem}{Remark}
\newtheorem{claim}{Claim}
\newtheorem{lem}{Lemma}
\newtheorem{prop}{Proposition}

\newtheorem{fact}{Fact}

\makeatother

\newcommand{\wt}{\widetilde}

\begin{document}
	\title{Age of Information-Reliability Trade-offs in Energy Harvesting Sensor Networks }

\author{Navid Nouri\footnote{Navid Nouri is with the School of Computer and Communication Sciences, EPFL, Lausanne, Switzerland (Email: navid.nouri@epfl.ch). Supported by ERC Starting Grant 759471.}, Darya Ardan, and Mahmood Mohassel Feghhi \footnote{Darya Ardan and Mahmood Mohassel Feghhi are with the Faculty of Electrical and Computer Engineering, University of Tabriz, Tabriz, Iran (Email: daryaardan@gmail.com, mohasselfeghhi@tabrizu.ac.ir).}}

	\maketitle
	\begin{abstract}
        Age of Information (AoI) is a recently defined quantity, which measures the freshness of information in a communication scheme. In this paper, we analyze a network that consists of a sensor node, an energy source and a receiver. The energy source is broadcasting energy and the sensor is charging its battery using energy-harvesting technologies. Whenever the battery gets fully charged, the sensor measures some quantity (called its status) from an environment, and (or) transmits its status to the receiver. The full analysis of AoI of this network, in the setting when each status is sent once, is given previously. However, that approach does not present a reliability guarantee better than the success probability of one transmission. In this paper, we present a closed form expression for the AoI of a deterministic and a randomized scheme that guarantee a desired probability of successful transmission for each status, alongside with a zero-error scheme. Furthermore, we define a novel notion called AoI-reliability trade-off and present the AoI-reliability trade-offs of our schemes. Additionally, we show that numerical results match our theoretical findings.  
	\end{abstract}
\section{Introduction}
\emph{Energy harvesting sensor networks} have become an attractive concept, due to multiple factors. First, and probably the most influential factor, is the necessity of energy harvesting for sensor networks used in a wide range of applications from medical sensors in/on the body, to sensors used in remote environments \cite{6449245}. The sensors with built-in (isolated) batteries are not desirable in such applications, since the sensor dies as soon as the battery is exhausted. Moreover, this technology enables us to save more energy, use renewable energy sources, and produce more compact and lighter devices, which are more desirable, specially in medical sensors. Energy harvesting, also received a lot of attention in the field of Internet of Things (IoT)-- see \cite{8389371} for a survey. Varshney \cite{4595260} introduced the concept of wireless power transmission for the first time, and works such as  \cite{6449245,6489506,FEGHHI2018108} studied the problem further for different architectures.

A natural metric for measuring the freshness of information in the receiver of communication systems, is defined as \emph{Age of Information} (AoI) \cite{AoI-main}. It is defined as the time elapsed since the generation of the freshest status update that has reached the receiver. The concept of AoI first appeared in \cite{6195689,10.1145/2034832.2034852,5984917}. There are numerous works investigating AoI from a queueing-theoretic standpoint, e.g., \cite{AoI-main,6310931,7364263, 8820073}. On the other hand, \cite{8254156, 8123937,age-minimal,8822722,8406846,8437904} study AoI of various energy harvesting systems and some propose  techniques for minimizing average AoI of their systems. We refer the reader to the survey \cite{yates2020age} for a more complete set of references.

Another main desired feature in communication systems is \emph{reliability}. In communication systems with higher reliability guarantees, messages are received by the receiver (without error) with higher probability. In most of the practical settings, designing a communication system with a high reliability guarantee is crucial.
 For instance, in medical sensors, it is vital that the measurements reach the monitor with the smallest possible error. In some other practical settings, such as streaming systems, it is important that each message gets to the receiver with some high probability. However, achieving high reliability guarantees are not for free, and results in schemes with higher average AoI. 
 
 \subsection{Our results}
 In this paper, we consider a system model where a sensor node with a energy harvesting technology transmits its status to the receiver over a noisy channel.\footnote{This system model is similar to what is used in \cite{krikidis2019average} and is described formally in Section~\ref{sec:sysmodel}.} We present two simple and natural schemes for this setting based on re-transmission of the failed messages. 
 
 Our first scheme is a {\bf deterministic scheme} (see Section~\ref{sec:det}), that considering the desired reliability guarantee, determines the maximum number of transmissions necessary for each status, and then transmits its status to the receiver, each time its battery is fully charged, until either the status is received by the receiver, or it reaches the maximum number of trials. In Theorem~\ref{thm:det} we present the average AoI of this scheme. Despite the fact that this scheme has the advantage of being a deterministic scheme, it is wasteful in certain regimes, when in order to compensate a small error probability, the sensor needs to transmit the status one more time (see the discussion at the end of Section~\ref{sec:det}). This issue also results in a non-smooth trade off between average AoI and reliability (see Figure~\ref{fig:tradeoff}). 
 
 Our second scheme is a {\bf randomized scheme} (see Section~\ref{sec:rand}). In this scheme, different than the deterministic scheme, the maximum number of transmissions for all statuses are not identical, and are decided in a random fashion. The average AoI of this scheme is presented in Theorem~\ref{thm:rand}. Our randomized approach, solves the wastefulness issue of the deterministic approach, and results in a smooth average AoI-reliability trade off (see Figure~\ref{fig:tradeoff}).
 
 Finally, as a direct corollary of our calculations, we present the average AoI of a zero-error scheme in Corollary~\ref{cor:zeroerror}.

\subsection{Related work}
We use a similar system model as in the recent work of Krikidis \cite{krikidis2019average}, with the following important difference that in our schemes, any status is received by the receiver with probability at least $1-\delta$ for any $\delta \in (0,1)$ that one desires (as opposed to a fixed value $\pi$ (success probability of one transmission) in \cite{krikidis2019average}). In other words, while the reliability guarantee of the approach of \cite{krikidis2019average} is determined by the channel, in our scheme the user determines the minimum probability of success for each status. Also, \cite{8006544} considers a communication system with unreliable multiaccess channels, where the sensor nodes try to transmit their status to the receiver by repeating the transmission, similar to our approach in this paper. However, their setting assumes that the nodes are connected to the power grid, as opposed to our energy harvesting setting. On the other hand, \cite{8422521} studies the trade off between energy efficiency and age of information. In their system model a server transmits information to multiple users
via multicasting based on requests from the users, where higher energy efficiency is achieved at the cost of higher age of information.

\subsection{Outline}
 In Section~\ref{sec:prelim}, we define basic notations and definitions used throughout the paper. Our system model is presented in Section~\ref{sec:sysmodel}. The main schemes and their analysis are presented in Section~\ref{sec:schemes}. Finally, the numerical results of our simulations and comparison of various scenarios are presented in Section~\ref{sec:simulations}. 
%\tableofcontents
\section{Preliminaries}\label{sec:prelim}
For any positive integer $k$, we define $[k]:=\{1,2,\ldots,k\}$ and we let $[0]:= \varnothing$. We use $\mathbf{1}_X$ as the indicator of event $X$, i.e., $\mathbf{1}_X=1$ if $X$ holds and $\mathbf{1}_X=0$, otherwise. We show any Bernoulli random variable $K$, which takes $1$ with probability $p$ and $0$ with probability $1-p$, using $K\sim \mathrm{Bern}(p)$.
For any $\pi\in(0,1)$, $F$ is a random variable with geometric distribution with parameter $\pi$, if for each positive integer $i$, $\Pr[F=i]=\pi (1-\pi)^{i-1}$,
and we denote it by $F \sim  \mathrm{Geom}(\pi)$. The fact below gives the expressions for the first and the second moments of a random variable with geometric distribution. 
\begin{fact}\label{fact:F}
	If $F\sim \mathrm{Geom}(\pi)$ for some $\pi \in (0,1)$, then we have
$
	\mathbb{E}[F]=\frac{1}{\pi}
$
	and
$
	\mathbb{E}\left[F^2\right] = \frac{2-\pi}{\pi^2}.
$
\end{fact}

\section{System model}\label{sec:sysmodel}
In this section, we present our system model, which is similar to the system model of \cite{krikidis2019average}, with some differences, which we highlight below. This model is a natural setting for energy harvesting systems. In this model, we have three nodes: {\bf (1)} an energy source node, $\mathcal{E}$, which broadcasts power in the environment, {\bf (2)} a sensor (or transmitter) node, $\mathcal{S}$, which uses the energy broadcasted by $\mathcal{E}$ to charge its battery and use this power to transmit its message later, {\bf (3)} a receiver node, $\mathcal{R}$, which receives the message transmitted by $\mathcal{S}$. There are three channels in this model, one from $\mathcal{E}$ to $\mathcal{S}$ for \emph{energy harvesting} (using RF signals), one from $\mathcal{S}$ to $\mathcal{R}$ for \emph{message transmission} and a feedback channel from $\mathcal{R}$ to $\mathcal{S}$, which is used to inform the sensor when a successful transmission happens. More specifically, the sensor node is not connected to the power grid and uses the power in its battery for the transmission. This battery is charged using the energy harvested from $\mathcal{E}$-$\mathcal{S}$ channel. The sensor node only transmits its status when its battery is fully charged and uses all the stored energy for the transmission. Sending information and harvesting energy both are possible simultaneously, due to the differences between frequency bands of the $\mathcal{E}$-$\mathcal{S}$ link and the $\mathcal{S}$-$\mathcal{R}$ link, and a proper architecture of battery \cite{luo2013optimal}. We consider times as tantamount size slots. In each time slot, the battery is getting charged in a quantized fashion in time. At the end of each time slot, the sensor can decide to transmit its status and (or) update its status. Below, we discuss the channels of this model.
\paragraph{Harvesting} The  $\mathcal{E}$-$\mathcal{S}$ channel is a Rayleigh  fading channel, which has power $n_t\sim \exp(\lambda)$ at time slot $t$.\footnote{We can also assume that an additive white Gaussian noise is present in this channel, however the contribution of this noise to the harvesting is negligible \cite{krikidis2019average}.} Assume that $\mathcal{E}$ broadcasts power $P$ in each time slot. The energy stored in the battery evolves as follows
\begin{align}
    E_t:=\min\{ \mathbf{1}_{E_{t-1}<B}E_{t-1}+\eta P n_t,B\},
\end{align}
where $\eta \in [0,1]$ is the RF-to-DC conversion rate and $B$ is the maximum capacity of the battery. Also, assume that the battery is out of charge in the beginning of the scheme, i.e., $E_{0}:=0$. Furthermore, whenever the sensor transmits a status, the transmission drains the whole energy in the battery. Also, for the ease of notation in the rest of the paper we let
\begin{align}
\beta:=\lambda B/(\eta P).
\end{align}
\paragraph{Transmission} The $\mathcal{S}$-$\mathcal{R}$ channel is also a Rayleigh fading channel, with power $m_t\sim \exp(\lambda)$, alongside the presence of an additive white Gaussian noise with variance $\sigma^2$. The signal-to-noise ratio for the $\mathcal{S}$-$\mathcal{R}$ channel at time slot $t$ is
\begin{align}
    \gamma_t=\frac{B\cdot m_t}{\sigma^2}.
\end{align}
Assuming the spectral efficiency of $r$ bits per channel use (BPCU), the probability of a successful transmission in the $\mathcal{S}$-$\mathcal{R}$ channel is given by 
\begin{align}
    \pi = \Pr\left[\log_2 (1+\gamma_t)> r\right] =\exp\left(-\lambda \frac{2^r-1}{B/\sigma^2}\right).
\end{align}
\paragraph{Feedback} The feedback channel is an error-free channel from the receiver to the sensor node, and is used to inform the sensor node whenever a successful transmission happens. The error-free assumption on this channel is reasonable due to the fact that the receiver is connected to the power grid and can transmit messages with high power. Note that since \cite{krikidis2019average} does not aim for any reliability guarantee, this channel is not needed for their setting.

 For the communication systems in general, if the most recently received message at time $t$ in the receiver is $m$, then AoI of the system at time $t$ is defined as the age of message $m$, i.e., the time passed from the generation of message $m$ in the transmitter. We consider the discrete definition of AoI, which defines AoI for the time slots (each has one time unit length) as opposed to continues values. 
 \begin{rem}\label{rem:AoI-disc}
 	Note that the discrete definition of AoI is similar to the continuous definition of AoI, which is more common in the literature, but we prefer the discrete one for our purpose, since this makes it clear that the transmission of the messages are happening in discrete times. Also, note that even using the continuous definition of AoI, all the results of this paper still hold, and this assumption does not ease the problem. 
 \end{rem}

\section{Schemes with reliability guarantees} \label{sec:schemes}
 In the following subsections, we first present and analyze a natural deterministic scheme, and then we use randomness to achieve an improved and smooth AoI-reliability trade off.
 
\subsection{Deterministic scheme}\label{sec:det}

Let $\pi$ be the probability that a single transmission over the channel is successful. Then, if the sensor transmits its current status $k$ times, where 
    $k:= \left\lceil\log_{1-\pi} \delta\right\rceil$,
then the status is received by the receiver with probability at least $1-\delta$. Basically, after each successful transmission or giving up on the previous status (we elaborate these cases below), the sensor waits for its battery to recharge. As soon as the battery gets fully charged, the sensor updates its status and uses its battery's energy to transmit the status to the receiver. We assume that there is one time slot allocated for the transmission, i.e., the receiver receives the transmitted message with a delay of one time slot.\footnote{The reader should note that even if the transmission was assumed to be immediate, it only affects the AoI calculations by a additive $-1$ factor, however, this assumption is also present in the system model of \cite{krikidis2019average} and is reasonable due to the delays caused by various factors such as the processing time in the receiver.} For any status sensed by the sensor, the sensor repeats transmitting it, until either the status is received without error by the receiver (\emph{successful transmission}), or the status reaches its limit for the maximum number of transmissions (\emph{giving up on the status}). 

Below, we present a formal version of our deterministic scheme in Algorithm~\ref{alg:k-LTS}.

\begin{algorithm}[H]
	\caption{\textsc{Deterministic Scheme}: A scheme that achieves $(1-\delta)$-Reliability deterministically.}  
	\label{alg:k-LTS} 
	\begin{algorithmic}[1]
		\State $k\gets \lceil \log_{1-\pi}\delta\rceil$ \label{line:k-det}
		\State Wait until the battery gets fully charged.  \Comment{Energy harvesting}\label{line:1-det}
		\If{the last transmission was successful or $\ell =k$}\newline
		\nonumber \text{~}\Comment{The feedback channel is used to determine the success of the last transmission}
		\newline \nonumber \text{~}\Comment{We assume that when we initialize the algorithm, $\ell=k$, so the if statement holds}
		\State Sense a new status and update $s$  
		\State $\ell\gets0$ 
		\EndIf
		\State Use the battery charge to send $s$ to the receiver 
		\State $\ell\gets \ell+1$ 
		\State Repeat
	\end{algorithmic}
\end{algorithm}

Now, we set a few notations, which we use frequently in the following claims and their proofs. For any integer $i\ge 1$, let $\tau_i$ be the time that $i$'th successful message is received by the receiver, and let $\tau_i'$ be the time that this message is generated (sensed) by the sensor. For any $i\ge 2$, we let $X_i$ denote the time between $(i-1)$'th and $i$'th successful transmission, i.e., $X_i:= \tau_i-\tau_{i-1}$. Also, let $H_i:= \tau_{i-1}-\tau'_{i-1}-1$. We also let $F_i$ denote the number of transmissions between $(i-1)$'th and $i$'th successful transmission, i.e., the number of full battery charges in that period. One can easily see that $F_i\sim \mathrm{Geom}(\pi)$, if $\pi$ is the probability that a single transmission over the channel is successful. For a better understanding of the upcoming proofs, we also define an auxiliary notion of AoI in the sensor, which measures the freshness of the message that the sensor has at any time. Figure~\ref{Fig:AoI} provides an instance of these definitions. Also, for ease of notation, we drop the indices when there is no ambiguity. 
 
 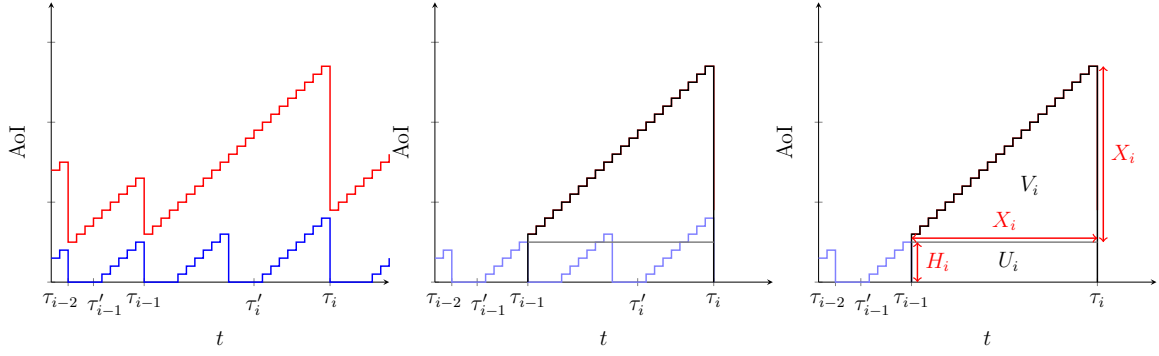
\begin{figure}[H]
 	\centering
 	\begin{subfigure}[b]{0.3\textwidth}
 		\centering
 		\begin{tikzpicture}[scale=0.7]
 		\pgfplotsset{width=8cm,compat=1.5}
 		\begin{axis}[%
 		,xlabel=$t$
 		,ylabel=$\mathrm{AoI}$
 		,ytick={}
 		,yticklabels={}
 		,axis x line = bottom,axis y line = left
 		,ymin=0
 		,ymax=35
 		,xmax= 40
 		,xmin= 0
 		,xtick={2,5,11,24,33}
 		,xticklabels={$\tau_{i-2}~~$,$~~~\tau_{i-1}'$,$\tau_{i-1}$,$\tau_i'$,$\tau_i$},
 		]
 		\addplot+[blue , const plot, no marks, thick] coordinates {(0,3) (1,4) (2,0) (3,0) (4,0) (5,0) (6,1) (7 ,2) (8,3) (9,4) (10 ,5) (11,0) (12,0) (13,0 ) (14,0) (15 ,1) (16,2) (17 ,3) (18 ,4) (19,5) (20, 6) (21,0 ) (22,0)(23 , 0) (24,0)  (25,1) (26,2) (27,3) (28,4) (29,5) (30,6) (31 ,7) (32,8) (33,0) (34 ,0) (35, 0) (36,0) (37,0) (38,1) (39,2) (40,3)} node[above,pos=.57,black] {};
 		\addplot+[red,const plot, no marks, thick] coordinates {(0,14) (1,15) (2,5) (3,6) (4,7) (5,8) (6,9) (7 ,10) (8,11) (9,12) (10 ,13) (11,6) (12,7) (13,8) (14,9) (15 ,10) (16,11) (17,12) (18,13) (19, 14) (20, 15) (21, 16) (22, 17) (23, 18) (24 , 19) (25,20) (26,21) (27,22) (28,23) (29,24) (30, 25) (31 , 26) (32, 27) (33,9) (34, 10) (35,11) (36,12) (37,13) (38,14) (39,15) (40,16) } node[above,pos=.57,black] {};

 		\end{axis}
 		
 		\end{tikzpicture}
 	\end{subfigure}
 	\begin{subfigure}[b]{0.3\textwidth}
 		\centering
 		\begin{tikzpicture}[scale=0.7]
 		\pgfplotsset{width=8cm,compat=1.5}
 		\begin{axis}[%
 		,xlabel=$t$
 		,ylabel=$\mathrm{AoI}$
 		,ytick={}
 		,yticklabels={}
 		,axis x line = bottom,axis y line = left
 		,ymin=0
 		,ymax=35
 		,xmax= 40
 		,xmin= 0
 		,xtick={2,5,11,24,33}
 		,xticklabels={$\tau_{i-2}~~$,$~~~\tau_{i-1}'$,$\tau_{i-1}$,$\tau_i'$,$\tau_i$},
 		]
 		
 		\addplot+[blue!50 , const plot, no marks, thick] coordinates {(0,3) (1,4) (2,0) (3,0) (4,0) (5,0) (6,1) (7 ,2) (8,3) (9,4) (10 ,5) (11,0) (12,0) (13,0 ) (14,0) (15 ,1) (16,2) (17 ,3) (18 ,4) (19,5) (20, 6) (21,0 ) (22,0)(23 , 0) (24,0)  (25,1) (26,2) (27,3) (28,4) (29,5) (30,6) (31 ,7) (32,8) (33,0) } node[above,pos=.57,black] {};
 		\addplot+[red,const plot, no marks, thick] coordinates { (11,6) (12,7) (13,8) (14,9) (15 ,10) (16,11) (17,12) (18,13) (19, 14) (20, 15) (21, 16) (22, 17) (23, 18) (24 , 19) (25,20) (26,21) (27,22) (28,23) (29,24) (30, 25) (31 , 26) (32, 27) (33,9)  } node[above,pos=.57,black] {};

 		\addplot+[black, const plot, no marks, thick] coordinates {(11,0) (11,6) (12,7) (13,8) (14,9) (15 ,10) (16,11) (17,12) (18,13) (19, 14) (20, 15) (21, 16) (22, 17) (23, 18) (24 , 19) (25,20) (26,21) (27,22) (28,23) (29,24) (30, 25) (31 , 26) (32, 27) (33,0)  } node[above,pos=.57,black] {};

 		\addplot+[gray, const plot, no marks, thick] coordinates { (11,5) (33 , 5) } node[above,pos=.57,black] {};

 		\end{axis}
 		\end{tikzpicture}
 	\end{subfigure}
 	\begin{subfigure}[b]{0.30\textwidth}
 		\centering
 		\begin{tikzpicture}[scale=0.7]
 		\pgfplotsset{width=8cm,compat=1.5}
 		\begin{axis}[%
 		,xlabel=$t$
 		,ylabel=$\mathrm{AoI}$
 		,ytick={}
 		,yticklabels={}
 		,axis x line = bottom,axis y line = left
 		,ymin=0
 		,ymax=35
 		,xmax= 40
 		,xmin= 0
 		,xtick={2,5,11,33}
 		,xticklabels={$\tau_{i-2}~~$,$~~~\tau_{i-1}'$,$\tau_{i-1}$,$\tau_i$},
 		]

 		\addplot+[blue!50 , const plot, no marks, thick] coordinates {(0,3) (1,4) (2,0) (3,0) (4,0) (5,0) (6,1) (7 ,2) (8,3) (9,4) (10 ,5) (11,0)} node[above,pos=.57,black] {};
 		\addplot+[red,const plot, no marks, thick] coordinates { (11,6) (12,7) (13,8) (14,9) (15 ,10) (16,11) (17,12) (18,13) (19, 14) (20, 15) (21, 16) (22, 17) (23, 18) (24 , 19) (25,20) (26,21) (27,22) (28,23) (29,24) (30, 25) (31 , 26) (32, 27) (33,9)  } node[above,pos=.57,black] {};

 		\addplot+[black, const plot, no marks, thick] coordinates {(11,0) (11,6) (12,7) (13,8) (14,9) (15 ,10) (16,11) (17,12) (18,13) (19, 14) (20, 15) (21, 16) (22, 17) (23, 18) (24 , 19) (25,20) (26,21) (27,22) (28,23) (29,24) (30, 25) (31 , 26) (32, 27) (33,0)  } node[above,pos=.57,black] {};

 		\addplot+[gray, const plot, no marks, thick] coordinates { (11,5) (33 , 5) } node[above,pos=.57,black] {};
 		
 		\node[above,black] at (225,5) {$U_i$};
 		\node[above,black] at (250,100) {$V_i$};
 		\addplot [<->, red, thick] coordinates { (11,5.5) (33 , 5.5) }node[above,pos=.5,red] {$X_i$};
 		\addplot [<->, red, thick] coordinates { (11.7,0) (11.7 , 5) }node[right,pos=.5,red] {$H_i$};
 		\addplot [<->, red, thick] coordinates { (33.7,5) (33.7 , 27) }node[right,pos=.5,red] {$X_i$};
 		\end{axis}
 		\end{tikzpicture}
 	\end{subfigure}
 	\caption{An instance for illustrating the notations defined above, and the steps that we use to calculate $\mathbb{E}[A_i]$.} 
 	\label{Fig:AoI}
 \end{figure}
First, we state the first and the second moments of $T$, where $T$ is the random variable for the time that it takes to fully charge an empty battery. 
\begin{prop}[Proposition 1 of \cite{krikidis2019average}]\label{prop:T}
	We have 
$
	\mathbb{E}[T]=1+\beta$ and 
$	\mathbb{E}[T^2]=1+3\beta+\beta^2$,
	where $\beta:=\lambda B/(\eta P)$ (see Section~\ref{sec:sysmodel} for the definition of parameters used in the definition of $\beta$).
\end{prop}
The proof of this proposition is given in \cite{krikidis2019average}.

 Now,  we state the following result about the first and the second moments of the time between $(i-1)$'th and $i$'th successful transmissions.
\begin{lem}\label{lem:XX2}
	Let $X:=X_i$ for any fixed $i\ge2$ and let  $F\sim \mathrm{Geom}(\pi)$. Then we have $\mathbb{E}[X]=\mathbb{E}[T]\cdot\mathbb{E}[F]$ 
 and
    $ \mathbb{E}[X^2]=\mathbb{E}\left[T^2\right] \mathbb{E}[F] + \mathbb{E}[T]^2\mathbb{E}\left[F^2\right]-\mathbb{E}[T]^2\mathbb{E}[F]$.
 \end{lem} 
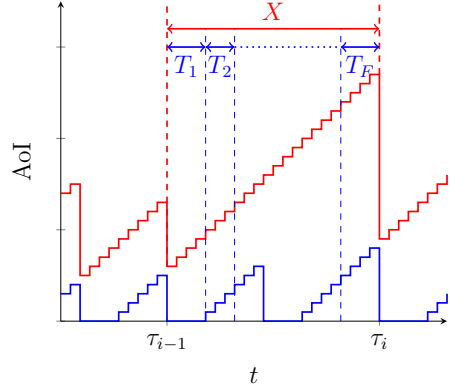
\begin{figure}[H]
	\centering
	\begin{tikzpicture}[scale=0.8]
	\pgfplotsset{width=8cm,compat=1.5}
	\begin{axis}[%
	,xlabel=$t$
	,ylabel=$\mathrm{AoI}$
	,ytick={}
	,yticklabels={}
	,axis x line = bottom,axis y line = left
	,ymin=0
	,ymax=35
	,xmax= 40
	,xmin= 0
	,xtick={11,33}
	,xticklabels={$\tau_{i-1}$,$\tau_i$},
	]
	\addplot+[blue , const plot, no marks, thick] coordinates {(0,3) (1,4) (2,0) (3,0) (4,0) (5,0) (6,1) (7 ,2) (8,3) (9,4) (10 ,5) (11,0) (12,0) (13,0 ) (14,0) (15 ,1) (16,2) (17 ,3) (18 ,4) (19,5) (20, 6) (21,0 ) (22,0)(23 , 0) (24,0)  (25,1) (26,2) (27,3) (28,4) (29,5) (30,6) (31 ,7) (32,8) (33,0) (34 ,0) (35, 0) (36,0) (37,0) (38,1) (39,2) (40,3)} node[above,pos=.57,black] {};
	\addplot+[red,const plot, no marks, thick] coordinates {(0,14) (1,15) (2,5) (3,6) (4,7) (5,8) (6,9) (7 ,10) (8,11) (9,12) (10 ,13) (11,6) (12,7) (13,8) (14,9) (15 ,10) (16,11) (17,12) (18,13) (19, 14) (20, 15) (21, 16) (22, 17) (23, 18) (24 , 19) (25,20) (26,21) (27,22) (28,23) (29,24) (30, 25) (31 , 26) (32, 27) (33,9) (34, 10) (35,11) (36,12) (37,13) (38,14) (39,15) (40,16) } node[above,pos=.57,black] {};
	
			\addplot+[dashed,red, const plot, no marks, thick] coordinates { (11,6) (11 , 35) } node[above,pos=.57,black] {};
			\addplot+[dashed,red, const plot, no marks, thick] coordinates { (33,27) (33 , 35) } node[above,pos=.57,black] {};

				\addplot [<->, red, thick] coordinates { (11,32) (33 , 32) }node[above,pos=.5,red] {$X$};
				
				\addplot [<->, blue, thick] coordinates { (11,30) (15 , 30) }node[below,pos=.5,blue] {$T_1$};
				
				\addplot+[dashed,blue, const plot, no marks] coordinates { (15,0) (15 , 32) } node[above,pos=.57,black] {};
				
				\addplot [<->, blue, thick] coordinates { (15,30) (18 , 30) }node[below,pos=.5,blue] {$T_2$};
				
				\addplot+[dashed,blue, const plot, no marks] coordinates { (18,0) (18 , 32) } node[above,pos=.57,black] {};
				
				\addplot [<->, blue, thick] coordinates { (29,30) (33 , 30) }node[below,pos=.5,blue] {$T_F$};
				
				\addplot+[dashed,blue, const plot, no marks] coordinates { (29,0) (29 , 32) } node[above,pos=.57,black] {};
				
					\addplot+[dotted,blue, const plot, no marks,thick] coordinates { (18,30) (29 , 30) } node[above,pos=.57,black] {};

	\end{axis}
	\end{tikzpicture}
	\caption{An illustration of the notations used in the proof of Lemma~\ref{lem:XX2}.}\label{fig:proofX}
\end{figure}
 \begin{proof}
	Let $F$ denote the random variable corresponding to the number of transmissions (or equivalently the number of battery recharges) between the $(i-1)$'th and the $i$'th successful transmissions. Also, for each $j\in [F]$, let $T_j$ denote the time that it takes for the $j$'th battery recharge. Now, if we define $X:=\tau_i-\tau_{i-1}$, then $X=T_1+T_2+\ldots+T_F$ (see Figure~\ref{fig:proofX}). Thus, for the expectation we have
	\begin{align}
	\mathbb{E}[X]&=\sum_{f=1}^{\infty}\mathbb{E}[T_1+T_2+\ldots+T_F|F=f]\Pr[F=f]\nonumber\\
	&=\sum_{f=1}^{\infty}\mathbb{E}[T_1+T_2+\ldots+T_f]\Pr[F=f]\nonumber\\
	&=\sum_{f=1}^{\infty} \mathbb{E}[T]\cdot f\cdot \Pr[F=f]\nonumber\\
	&=\mathbb{E}[T] \mathbb{E}[F].
	%&=\frac{\beta+1}{\pi},
	\end{align}
	Similarly, we have 
	\begin{align}
	\mathbb{E}\left[X^2\right]&=\sum_{f=1}^{\infty}\mathbb{E}\left[\left(T_1+T_2+\ldots+T_F\right)^2|F=f\right]\Pr[F=f]\nonumber\\
	&=\sum_{f=1}^{\infty}\mathbb{E}\left[\left(T_1+T_2+\ldots+T_f\right)^2\right]\Pr[F=f]\nonumber\\
	&= \sum_{f=1}^{\infty}\left(\mathbb{E}\left[T^2\right]\cdot f + \mathbb{E}[T]^2 \cdot (f^2-f)\right)\Pr[F=f]\nonumber\\
	&=\mathbb{E}\left[T^2\right] \mathbb{E}[F] + \mathbb{E}[T]^2\mathbb{E}\left[F^2\right]-\mathbb{E}[T]^2\mathbb{E}[F].
	\end{align}
\end{proof}

We denote the area under the AoI curve from $(i-1)$'th successful transmission to the $i$'th successful transmission by $A_i$. We want to calculate the expected value of $A_i$ for any $i\ge2$. Fix $i$, and for simplicity of notation let $A:=A_i$. As in Figure~\ref{Fig:AoI}, one can see that the area corresponding to $A$ can be divided into a rectangle  and a triangle\footnote{Note that this region is not  exactly a triangle because of discrete time slots that we use (see Remark~\ref{rem:AoI-disc})}, which their area are denoted by $U$ and $V$, respectively. Note that by linearity of expectation, we have 
$    \mathbb{E}[A]=\mathbb{E}[U]+\mathbb{E}[V]$.

The following lemma gives the expression for $\mathbb{E}[V]$.
\begin{lem} \label{lem:V}
$
\mathbb{E}[V]=\frac{1}{2}\left(\mathbb{E}[X^2]+\mathbb{E}[X]\right).
$
\end{lem}
\begin{proof}
Let $X$ be the random variable that corresponds to the amount of time between $(i-1)$'th successful transmission and $i$'th successful transmission, i.e., $X:=\tau_i-\tau_{i-1}$. Then by the definition of $V$, we have the following:
$
    V=\sum_{j=1}^{X} j = \frac{X(X+1)}{2},
$
which translates to
$
    \mathbb{E}[V]=\mathbb{E}\left[ \frac{X(X+1)}{2}\right]= \frac{1}{2}\left(\mathbb{E}[X^2]+\mathbb{E}[X]\right).
$
This concludes the proof. 
\end{proof}
Now we calculate the expected area of the rectangle part, i.e., $\mathbb{E}[U]$. As per Figure~\ref{Fig:AoI}, one should note that the random variables corresponding to the vertical and horizontal sides of the rectangle are independent of each other, since each transmission is successful with probability $\pi$, independent of the previous transmissions. As mentioned before, the random variable corresponding to the horizontal side is $X$, and we denote the random variable corresponding to the length of the vertical side by $H$. Suppose that $T'_1, T'_2, \ldots, T'_z$ are the random variables corresponding to the time for each battery recharge used for transmitting $m_{i-1}$ (the $(i-1)$'th successful message), where $z$ is an integer such that $z$'th transmission of $m_{i-1}$ is successful (and the previous ones have failed). Now, since the sensor updates its status when its battery is recharged for the first time (after a successful transmission or after giving up on the previous status), then  $H_i:=T'_2+\ldots+T'_z$. Thus we have
\begin{align}\label{eq:U}
    \mathbb{E}[U]=\mathbb{E}[H\cdot X]=\mathbb{E}[H]\cdot\mathbb{E}[X],
\end{align}
since $H$ and $X$ are independent random variables. 
In the following lemma, we calculate $\mathbb{E}[H]$.
\begin{lem}\label{lem:Hdet}
If we use the deterministic scheme (Algorithm~\ref{alg:k-LTS}), then we have 
\begin{align}
    \mathbb{E}[H]= \mathbb{E}[T]\cdot \left(\frac{1}{\pi}-\frac{k(1-\pi)^k}{1-(1-\pi)^k}-1\right).
\end{align}
\end{lem}
\begin{proof}
For a fixed $i\ge2$, let $m$ be the $(i-1)$'th successful message and let $Y$ be a random variable such that $Y\sim \mathrm{Geom}(\pi)$. Now let random variable $Z$ be such that
\begin{align}
    \forall j\in [k]: \Pr[Z=j]=\Pr[Y=j|Y\le k].
\end{align} 
One can see that $Z$ is the distribution of number of transmissions of message $m$, since the condition $Y\le k$ guarantees that this message is successful. Now, let $T'_1$, $T'_2$, \ldots, $T'_Z$ denote the amount of time that it takes for $Z$ full battery charges used in the transmissions of $m$. Now, we can calculate the expected value of $H$ as follows:
\begin{align}
    \mathbb{E}[H] &= \mathbb{E}[T'_2+\ldots+T'_Z]\nonumber\\
    &= \sum_{j=1}^{k} \mathbb{E}[T'_2+\ldots+T'_Z|Z=j]\Pr[Z=j]\nonumber\\
     &= \sum_{j=1}^{k} \mathbb{E}[T'_2+\ldots+T'_j]\Pr[Y=j|Y\le k]\nonumber\\
     &=\mathbb{E}[T] \cdot \left(\sum_{j=1}^{k} (j-1)\cdot \Pr[Y=j|Y\le k]\right)\nonumber\\
     &=\mathbb{E}[T] \cdot \left(\sum_{j=1}^{k} (j-1)\cdot \frac{(1-\pi)^{j-1}\pi}{1-(1-\pi)^{k}}\right)\nonumber\\
     &= \mathbb{E}[T]\cdot \left(\frac{1}{\pi}-\frac{k(1-\pi)^k}{1-(1-\pi)^k}-1\right),
\end{align}
where the last transition is by Claim~\ref{claim:arith} (see Appendix~\ref{sec:omitted}).
\end{proof}
Now, we can combine the results and get the following result on the AoI of the deterministic scheme. 
\begin{thm}\label{thm:det}
For any $\delta> 0$ such that $\delta\le 1-\pi$, there exist a deterministic scheme which achieves the following average AoI
\begin{align}
    \mathrm{AoI}_{\mathrm{det}}=(1+\beta)\left(\frac{2}{\pi}-\frac{k(1-\pi)^k}{1-(1-\pi)^k}-\frac{3}{2}\right)+\frac{2\beta+1}{2(1+\beta)},
\end{align}
where $k=\lceil \log_{1-\pi}\delta\rceil$. Furthermore, any status is received by the receiver with probability at least $1-\delta$.
\end{thm}
\begin{proof}
First, we prove the reliability guarantee. Note that a status is not received by the receiver if it fails $k$ transmissions, where $k=\lceil \log_{1-\pi}\delta\rceil$ (see Line~\ref{line:k-det} of Algorithm~\ref{alg:k-LTS}). Since the probability of failure of each transmission is $1-\pi$ then for any status $s$, 
\begin{align}
    \Pr[\text{status $s$ fails}]=(1-\pi)^k \le \delta,
\end{align}
where the last inequality holds since $k=\lceil \log_{1-\pi}\delta\rceil$.

Now, we prove the expression for Average AoI of this scheme. Note that we have
$$
    \mathrm{AoI}_{\mathrm{det}}= \frac{\mathbb{E}[A]}{\mathbb{E}[X]}
    =\frac{\mathbb{E}[V]+\mathbb{E}[U]}{\mathbb{E}[X]}
    =\frac{\frac{1}{2}\left(\mathbb{E}[X^2]+\mathbb{E}[X]\right)+\mathbb{E}[U]}{\mathbb{E}[X]},
$$
where the last equality is due to Lemma~\ref{lem:V} (and see the discussion in \cite{AoI-main} for the standard argument of the first equality).
Then, we have
\begin{align}
    \mathrm{AoI}_{\mathrm{det}}&=\frac{1}{2}\left(\frac{\mathbb{E}[X^2]}{\mathbb{E}[X]}+1\right)+\mathbb{E}[T]\left(\frac{1}{\pi}-\frac{k(1-\pi)^k}{1-(1-\pi)^k}-1\right)\nonumber\\
    &=\left(\frac{2-\pi}{2\pi}\right)(\beta+1)+\frac{2\beta+1}{2(1+\beta)}+\left(1+\beta\right)\left(\frac{1}{\pi}-\frac{k(1-\pi)^k}{1-(1-\pi)^k}-1\right)\nonumber\\
    &=(1+\beta)\left(\frac{2}{\pi}-\frac{k(1-\pi)^k}{1-(1-\pi)^k}-\frac{3}{2}\right)+\frac{2\beta+1}{2(1+\beta)},
\end{align}
where the first equality is by Lemma~\ref{lem:Hdet} and \eqref{eq:U}. The second equation is due to Proposition~\ref{prop:T}, Lemma~\ref{lem:XX2} and Fact~\ref{fact:F}, and the third equality is by reordering the terms. 
\end{proof}
\begin{rem}
Theorem~\ref{thm:det} recovers the result of \cite{krikidis2019average}, as a special case, when one sets $k=1$. Recall that in that case, the reliability guarantee of the system is $\pi$.  
\end{rem}

One can see Theorem~\ref{thm:det} as a trade off between reliability and average AoI. Moreover, note that the blue curve in Figure~\ref{fig:tradeoff}, which presents the trade off curve for the setting where $\beta=87$ and $\pi=0.65$, has discontinuities. The deterministic approach of setting a fix $k$ as maximum number of attempts, in order to achieve error probability of at most $\delta$, is wasteful. More specifically, the source of inefficiency in Algorithm~\ref{alg:k-LTS} comes from the rounding in Line~\ref{line:k-det} of Algorithm~\ref{alg:k-LTS}. Note that for $k=\lceil \log_{1-\pi}\delta\rceil $, we have that
$
    (1-\pi)^{k} \le \delta < (1-\pi)^{k-1}.
$

In the next subsection, we modify the scheme and use randomness to achieve a better average AoI with error probability of at most $\delta$ for each status, which provides us smooth average AoI-reliability trade offs.

\subsection{Randomized scheme}\label{sec:rand}
In Algorithm~\ref{alg:rand}, we use randomness in determining the maximum number of trials for each status. For each status $s$, we will choose this limit, $\wt{k}$, with proper probability between $k$ and $k-1$, so that the error probability becomes $\delta$ (see Line~\ref{line:bern} in Algorithm~\ref{alg:rand}). Needless to say that when $k=1$, error probability of $1-\pi$ cannot be increased to $\delta$ in the hope of getting better AoI, as any status has to be sent at least once. Below we present our randomized scheme in Algorithm~\ref{alg:rand}. 

\begin{algorithm}[H]
	\caption{\textsc{Randomized Scheme}: A scheme that achieves $(1-\delta)$-Reliability in a randomized fashion.}  
	\label{alg:rand} 
	\begin{algorithmic}[1]
		\State Wait until the battery gets fully charged. \Comment{Energy harvesting} \label{line:1-rand}

		\If{the last transmission was successful or $\ell =\wt{k}$}\newline
		\nonumber \text{~}\Comment{The feedback channel is used to determine the success of the last transmission}
		\newline \nonumber\text{} \Comment{We assume that when we initialize the algorithm $\ell=\wt{k}=0$, so the if statement holds}
		\State $\ell\gets 0$
		\State Sense a new status and update status $s$
		\State $k\gets \lceil \log_{1-\pi}\delta\rceil$ \label{line:k-rand}
		\If {$k>1$} 
		\State $p_1 \gets  (1-\pi)^{k}$,~~$p_2 \gets  (1-\pi)^{k-1}$
		
		\State $\wt{k} \gets  \mathrm{Bern}\left(\frac{\delta-p_2}{p_1-p_2}\right)+(k-1)$
		\label{line:bern}
		
		\Else 
		\State $\wt{k}\gets 1$ \label{line:ktilde1} 
		\EndIf
		
		\EndIf 
		\State Use the battery charge to send $s$ to the receiver \label{line:send-rand}
		\State $\ell \gets \ell+1$ 
		\State Repeat
	\end{algorithmic}
\end{algorithm}

Now, we calculate the average AoI of the scheme presented in Algorithm~\ref{alg:rand}. Suppose that $A:=A_i$ is the area under the AoI-time curve between $\tau_{i-1}$ and $\tau_i$ as described in the previous section. Again, the same as the previous section, we divide $A$ into two parts $U$ and $V$, where $U$ represents the area of the rectangle part, and $V$ is the area of the triangle part. One can easily check that since the area of $V$ only depends on the time that it takes to send any successful message to the receiver, the exact same result that we had for $\mathbb{E}[V]$ in the previous section holds, i.e.,
\begin{align}
    \mathbb{E}[V]=\frac{1}{2}\left(\mathbb{E}[X^2]+\mathbb{E}[X]\right).
\end{align}

Now, we analyze $\mathbb{E}[U]$, which needs different techniques for this scheme compared to the scheme of the previous section, and gives the improvement. 

Here, it is useful to remind the reader about the randomized scheme that we use. Our algorithm after any successful transmission, updates its status and then in Line~\ref{line:bern} of Algorithm~\ref{alg:rand} decides an upper-bound on the number of trials of this status, and if the number of trials exceeds that limit, then updates its status and then again chooses an independent upper-bound on the number of trials for this message  and so on. As mentioned above, $\mathbb{E}[X]$ is exactly the same as previous section, however, we need to calculate $\mathbb{E}[H]$ for this scheme. The next lemma is the main technical challenge of this section. 

\begin{lem}\label{lem:rand_H}
 For the randomized scheme (Algorithm~\ref{alg:rand}) if we let $\alpha:=\frac{\delta-p_2}{p_1-p_2}$, then for any $i>2$, we have 
 $
     \mathbb{E}[H_i]=p\cdot h_1 + (1-p)\cdot h_2
$
 where 
 \begin{itemize}
     \item $p=\frac{\alpha\left(1-(1-\pi)^k\right)}{1-\alpha(1-\pi)^k-(1-\alpha)(1-\pi)^{k-1}}$,
     \item $h_1=\mathbb{E}[T] \cdot \left(\frac{1}{\pi}-\frac{k(1-\pi)^k}{1-(1-\pi)^k}-1\right)$,
     \item $h_2=\mathbb{E}[T] \cdot \left(\frac{1}{\pi}-\frac{(k-1)(1-\pi)^{k-1}}{1-(1-\pi)^{k-1}}-1\right)$.
 \end{itemize}
\end{lem}
\begin{proof}
First, let $m'$ be the $(i-1)$'th successful message and let $k'$ be the maximum number of repetitions determined by Line~\ref{line:bern} of Algorithm~\ref{alg:rand} for message $m'$. Note, either $k'=k$ or $k'=k-1$. Also, let $\tau:=\tau'_{i-1}$ and $H:=H_i$ for ease of notation. Now, we can write:
\begin{align}
    \mathbb{E}[H] &= \Pr[k'=k]\cdot\mathbb{E}[H | k' = k]  +\Pr[k'=k-1]\cdot \mathbb{E}[H|k'=k-1]\label{eq:expHcond}
\end{align}
Let $Y$ be a random variable such that $Y\sim \text{Geom}(\pi)$. Then, we have 
\begin{align}
    \mathbb{E}[H|k'=k] &=\sum_{i=1}^{k}\mathbb{E}[T_2+\ldots+T_Y|Y=i,k'=k]\Pr[Y=i| Y \le k' ,k'=k]\nonumber\\
    &=\sum_{i=1}^{k} \mathbb{E}[T_2+\ldots+T_i]\Pr[Y=i| Y \le k]\nonumber\\
    &= \sum_{i=1}^{k} \mathbb{E}[T]\cdot (i-1)\cdot\Pr[Y=i| Y \le k]\nonumber\\
    &= \mathbb{E}[T]\cdot \sum_{i=1}^{k}  (i-1)\cdot \frac{(1-\pi)^{i-1}\pi}{1-(1-\pi)^k}\nonumber\\
    &=\mathbb{E}[T] \cdot \left(\frac{1}{\pi}-\frac{k(1-\pi)^k}{1-(1-\pi)^k}-1\right),
\end{align}
where the last transition is by Claim~\ref{claim:arith} (see Appendix~\ref{sec:omitted}). 
Similarly, we have
\begin{align}
    \mathbb{E}[H|k'=k-1] 
    &=\mathbb{E}[T] \cdot \left(\frac{1}{\pi}-\frac{(k-1)(1-\pi)^{k-1}}{1-(1-\pi)^{k-1}}-1\right).
\end{align}
If we let $h_1:=\mathbb{E}[T] \cdot \left(\frac{1}{\pi}-\frac{k(1-\pi)^k}{1-(1-\pi)^k}-1\right)$ and $h_2:=\mathbb{E}[T] \cdot \left(\frac{1}{\pi}-\frac{(k-1)(1-\pi)^{k-1}}{1-(1-\pi)^{k-1}}-1\right)$ then using \eqref{eq:expHcond} we have
$
    \mathbb{E}[H]= \Pr[k'=k]\cdot h_1+ \Pr[k'=k-1]\cdot h_2.
$

Now it remains to calculate $\Pr[k'=k]$, which is the probability that the scheme decides to upper bound the maximum transmission of message $m'$ by $k$ (and not $k-1$). One should note the subtle point here, which is this probability is different from the probability of setting $\wt{k}$ for each message in the algorithm. The reason is, it is conditioned that message $m'$ is successful. In other words, suppose that after successfully sending the $(i-2)$'th message, the sensor updates its status to $m_1$ and sets $\wt{k}_1:=\wt{k}(m_1)$ in Line~\ref{line:bern} of Algorithm~\ref{alg:rand}. Now, let $W=1$ if message $m_1$ is successful and $W=0$ otherwise. Thus,
\begin{align}
    \Pr[k'=k]= &\Pr[k'=k|\wt{k}_1=k, W=1] \Pr[\wt{k}_1=k, W=1]\nonumber\\
    &+\Pr[k'=k|\wt{k}_1=k, W=0] \Pr[\wt{k}_1=k, W=0]\nonumber\\
    &+\Pr[k'=k|\wt{k}_1=k-1, W=1] \Pr[\wt{k}_1=k-1, W=1]\nonumber\\
    &+\Pr[k'=k|\wt{k}_1=k-1, W=0] \Pr[\wt{k}_1=k-1, W=0].
\end{align}
Furthermore, since geometric distribution is memory-less 
\begin{align}
     \Pr[k'=k]= &1\cdot \Pr[\wt{k}_1=k, W=1]
    +\Pr[k'=k]\cdot\left(\Pr[\wt{k}_1=k, W=0]
    + \Pr[\wt{k}_1=k-1, W=0]\right)\nonumber\\
    = & \alpha\left(1-(1-\pi)^k\right) + \Pr[k'=k]\cdot \left(\alpha\cdot (1-\pi)^k+ (1-\alpha)(1-\pi)^{k-1}\right).
\end{align}
Thus, we have 
$
    \Pr[k'=k]=\frac{\alpha\left(1-(1-\pi)^k\right)}{1-\alpha(1-\pi)^k-(1-\alpha)(1-\pi)^{k-1}}
$
where $\alpha=\Pr[\wt{k}=k]$ for any message as assumed in the statement of the lemma. Now, if we let $p:=\frac{\alpha\left(1-(1-\pi)^k\right)}{1-\alpha(1-\pi)^k-(1-\alpha)(1-\pi)^{k-1}}$ then the claim holds. 
\end{proof}
\begin{thm}\label{thm:rand}
For any $\delta> 0$ such that $\delta\le (1-\pi)^2$, the randomized scheme (Algorithm~\ref{alg:rand}) achieves the following average AoI
 \begin{align}
    \mathrm{AoI}_{\mathrm{rand}}=\frac{1}{2}\left( \frac{\mathbb{E}[X^2]}{\mathbb{E}[X]}+1\right)+\mathbb{E}[H],
\end{align}
where $\mathbb{E}[H]$ is calculated in Lemma~\ref{lem:rand_H}, and $\mathbb{E}[X]$ and $\mathbb{E}[X^2]$ are calculated in Lemma~\ref{lem:XX2}.
Furthermore, any status is received by the receiver with probability at least $1-\delta$.
\end{thm}
\begin{proof}
	First, we start by proving the reliability guarantee. We consider two cases: 
	
	{\bf Case I ($k>1$):} Recall that $p_1=(1-\pi)^k$ and $p_2=(1-\pi)^{k-1}$ as per Algorithm~\ref{alg:rand}. In this case $\wt{k}=k$ with probability $\frac{\delta-p_2}{p_1-p_2}$ and $\wt{k}=k-1$, otherwise (see Line~\ref{line:bern} in Algorithm~\ref{alg:rand}). Let $W:=\mathbf{1}_{\wt{k}=k}$, be the indicator random variable for the event of $\wt{k}=k$. Then
	\begin{align}
	\Pr[\text{fail}]&=\Pr[W=1]\Pr[\text{fail}|W=1]+\Pr[W=0]\Pr[\text{fail}|W=0]\nonumber\\
	&=\left(\frac{\delta-p_2}{p_1-p_2}\right)(1-\pi)^k + \left(1-\frac{\delta-p_2}{p_1-p_2}\right)(1-\pi)^{k-1}\nonumber\\
	&=\left(\frac{\delta-p_2}{p_1-p_2}\right)\cdot p_1 + \left(1-\frac{\delta-p_2}{p_1-p_2}\right)\cdot p_2\nonumber\\
	&= \delta.
	\end{align}
	{\bf Case II ($k=1$):} In that case Algorithm~\ref{alg:rand} do not use randomness (see Line~\ref{line:ktilde1} of Algorithm~\ref{alg:rand}) and we have
	$
	\Pr[\text{fail}]=1-\pi\le \delta,
	$
	since $\log_{1-\pi}\delta \le 1$.
	
	So, in both cases we proved that the failure probability is at most $\delta$.
	
	The proof of the AoI value is immediate by the definition of $X$ and the definition of $H$.
\end{proof}
\begin{figure}

  	  \centering{\includegraphics[scale=0.4]{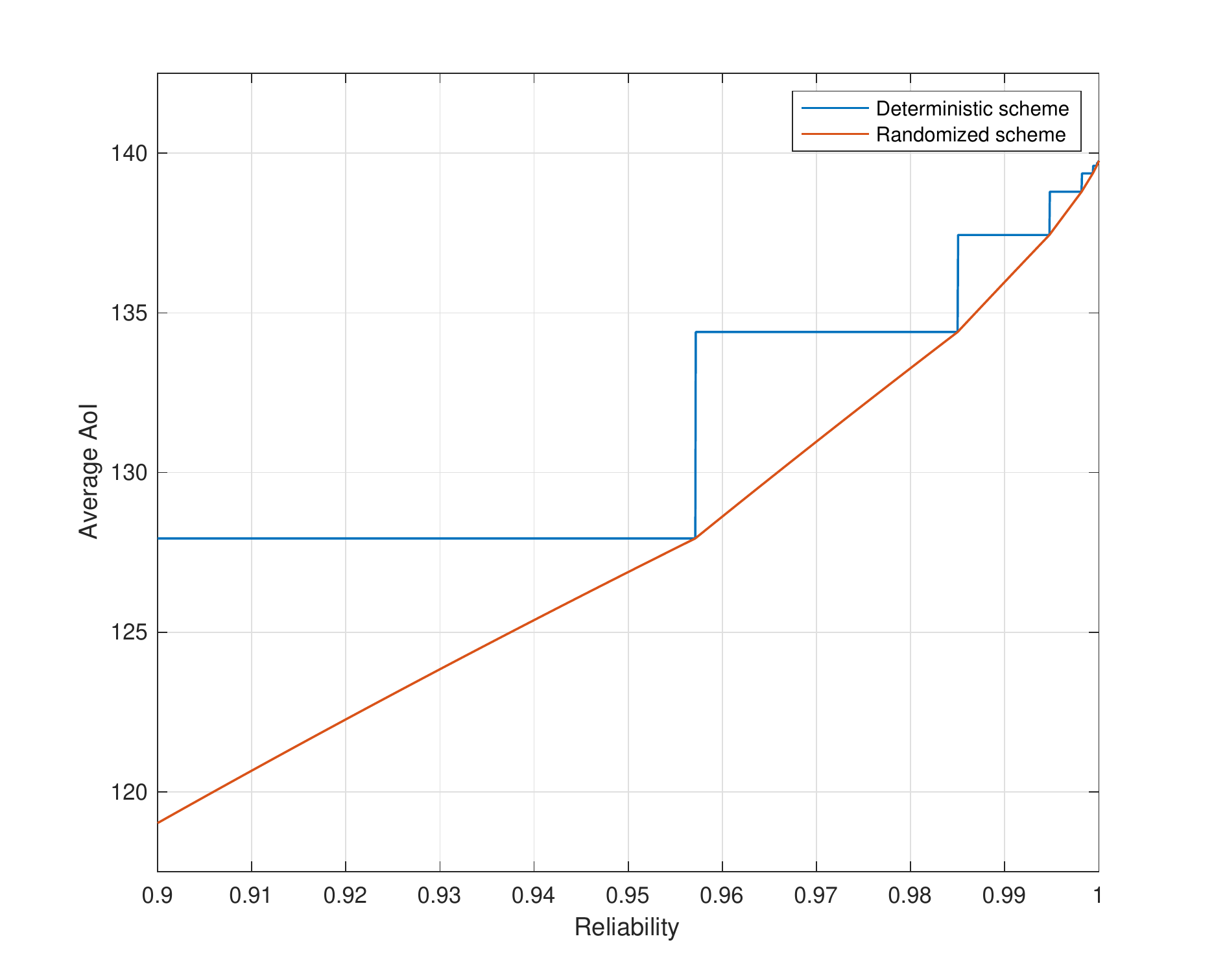}}

\caption{Comparison of average AoI-reliability trade offs for the setting where $\beta=87$ and $\pi=0.65$, for the deterministic and the randomized schemes.}\label{fig:tradeoff}
\end{figure}
In Figure~\ref{fig:tradeoff}, we compare the average AoI-reliability trade off curve of randomized scheme to the deterministic scheme. One can clearly see that using randomness we improved the average AoI of the system and achieved a smooth trade off curve. 

\subsection{Zero-error scheme}\label{sec:zero-error}
In order to achieve a zero-error scheme, we can set $\delta=0$, or equivalently set the maximum number of trials, $k$, to be infinity. Consequently, the sensor insists on sending each status until a successful transmission happens. The reader should note that in this case there is no difference between the deterministic and the randomized schemes. The average AoI of a zero-error scheme is stated below. 
\begin{cor}[Average AoI of zero-error scheme]\label{cor:zeroerror}
	The average AoI of a zero-error scheme is
$
(1+\beta)\cdot \left(\frac{4-3\pi}{2\pi}\right)+\frac{2\beta+1}{2(1+\beta)}.
$
\end{cor} 
\begin{proof}
	\begin{align}
	\lim_{k \rightarrow \infty} \text{AoI}_{\text{det}}&=\lim_{k\rightarrow \infty} (1+\beta)\left(\frac{2}{\pi}-\frac{k(1-\pi)^k}{1-(1-\pi)^k}-\frac{3}{2}\right)+\frac{2\beta+1}{2(1+\beta)}\nonumber\\
	&= (1+\beta)\cdot \left(\frac{4-3\pi}{2\pi}\right)+\frac{2\beta+1}{2(1+\beta)}.
	\end{align}
\end{proof}
The difference between the average AoI of the zero-error scheme and the simple case when each status is transmitted only once (as in \cite{krikidis2019average}) can be seen as the cost of achieving a zero-error scheme. Interestingly, in some regimes, the difference is relatively small, i.e., one can achieve a zero-error scheme with only a relatively small increase in the average AoI of the system. 

\section{Simulations}\label{sec:simulations}

In this section, we present the results of the simulations for the randomized scheme. We use a similar setting for simulations as \cite{krikidis2019average}: we set $\sigma^2=-50 \mathrm{dBm}$, $\eta=0.5$ and $r=0.05 \mathrm{BPCU}$ and we also assume that the sensor node has distance $d$ meters away from both the energy transmitter and the receiver. We also model the power gain of the channels as $\lambda=10^3d^\alpha$, where $\alpha=2.2$ is the path-loss exponent \cite{8275025}. 

First, we compare the average AoI of the randomized scheme and the deterministic scheme, with the same parameters. In Table~\ref{tab:det-rand}, we present the theoretical results and the outcome of simulations, for the case when $d=20m$ and $P=1W$, and for two different reliability guarantees, $80\%$ and $90\%$. First observation is that, as expected, generally using the randomized scheme lowers the average AoI. In Table~\ref{tab:det-rand}, the second observation is that even though the desired reliability in the first two rows are $90\%$ and the desired reliability of the last two rows are $80\%$, the deterministic scheme provides the same average AoI. The reason is the wasteful approach of this scheme. More specifically, the algorithm uses the same number $k$, for the limit of maximum number of transmission for both cases (see the discussion in the end of Section~\ref{sec:det}). However, the randomized scheme does not suffer from this phenomenon, due to use of randomness.

\begin{table}[H]
	\centering
	\begin{tabular}{c|c|cc|cc}
		\multicolumn{2}{c}{$ $} &\multicolumn{4}{|c}{Average AoI (in time slots)}\\
		\hline
		\multicolumn{2}{c}{Parameters} &\multicolumn{2}{|c}{Deterministic scheme}&\multicolumn{2}{|c}{Randomized scheme }  \\
		\hline
		Battery capacity& $\delta$ & Theory &  Simulation    & Theory & Simulation  \\
		\hline
		$1 \times 10^{-3}$ &$0.1$& $\approx 1425.6$ & $\approx 1437.7$  			& $\approx 1361.2$   				& $\approx 1364.5$\\
		$1.5 \times 10^{-3}$ &$0.1$& $\approx 1799.1$ & $\approx 1797.2$  			& $\approx 1641.3$   				& $\approx 1641.7$\\
		$1 \times 10^{-3}$ &$0.2$& $\approx 1425.6$ & $\approx 1421.6$  			& $\approx 1204.7$   				& $\approx 1204.0$\\
		$1.5 \times 10^{-3}$ &$0.2$& $\approx 1799.1$ & $\approx 1804.1$  			& $\approx 1502.0$   				& $\approx 1502.6$\\
		
		\hline
	\end{tabular}
	\caption{Comparison of average AoI of the deterministic and the randomized scheme, when $d=20m$, $P=1W$, for various $\delta$ values.}
	\label{tab:det-rand}
\end{table}

For the next set of simulations, in the randomized scheme, we set $\delta=0.01$, which implies $99\%$ reliability guarantee. In figure~\ref{fig:rand}, we present 5 curves for $$(d,P)\in \{(20 ,1),(15,1),(20,3),(20,5),(20,10)\},$$ where $P$ represent the power of energy transmitter in watts and $d$ is the distance of sensor node from the energy transmitter and the receiver in meters, as explained above.  
\begin{figure}[H]
	\centering
	\includegraphics[scale=0.4]{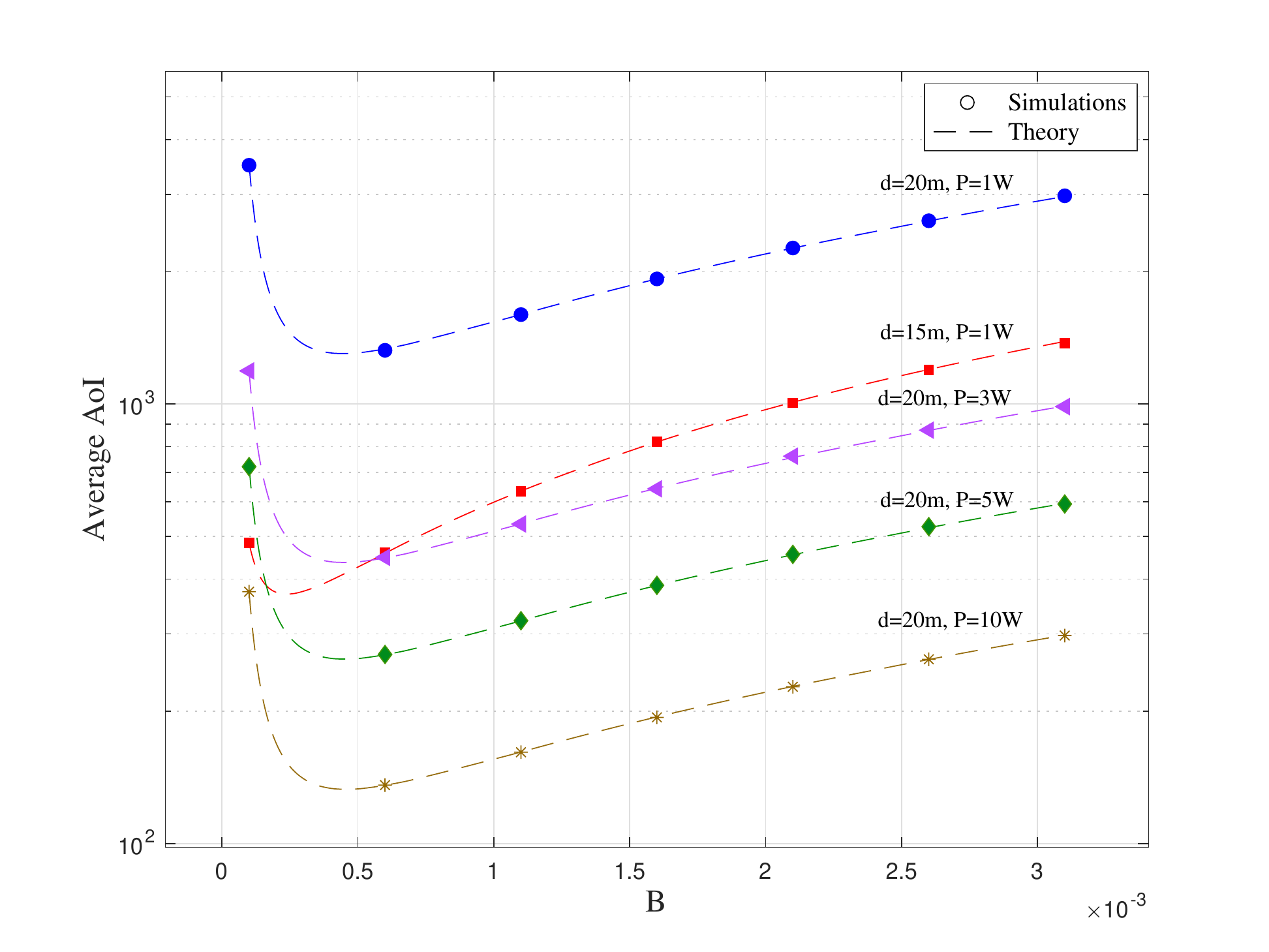}

	\caption{Average AoI with reliability guarantee of $99\%$ for a range of battery capacity for five different settings of parameters.}\label{fig:rand}
\end{figure}

For the set of simulations illustrated in Figure~\ref{fig:delta}, we set $P=1W$, $d=20m$ and other parameters as above in the randomized scheme. We compare average AoI curves for $\delta\in\{1,0.1,0.01,0\}$. One should note that $\delta=1$ means that reliability guarantee is equal to $\pi$, and $\delta=0$ correspond to the zero-error scheme (see Section~\ref{sec:zero-error}). This figure shows that as we aim for higher reliability guarantees, the average AoI of the system increases, as expected. Furthermore, in Table~\ref{tab:rel}, we compare the reliability promised by setting $\delta$ and Theorem~\ref{thm:rand} (theory), and the reliability we get using simulations for various battery capacities, when $P=1W$ and $d=20m$.

\begin{figure}[H]
	\centering
	\includegraphics[scale=0.4]{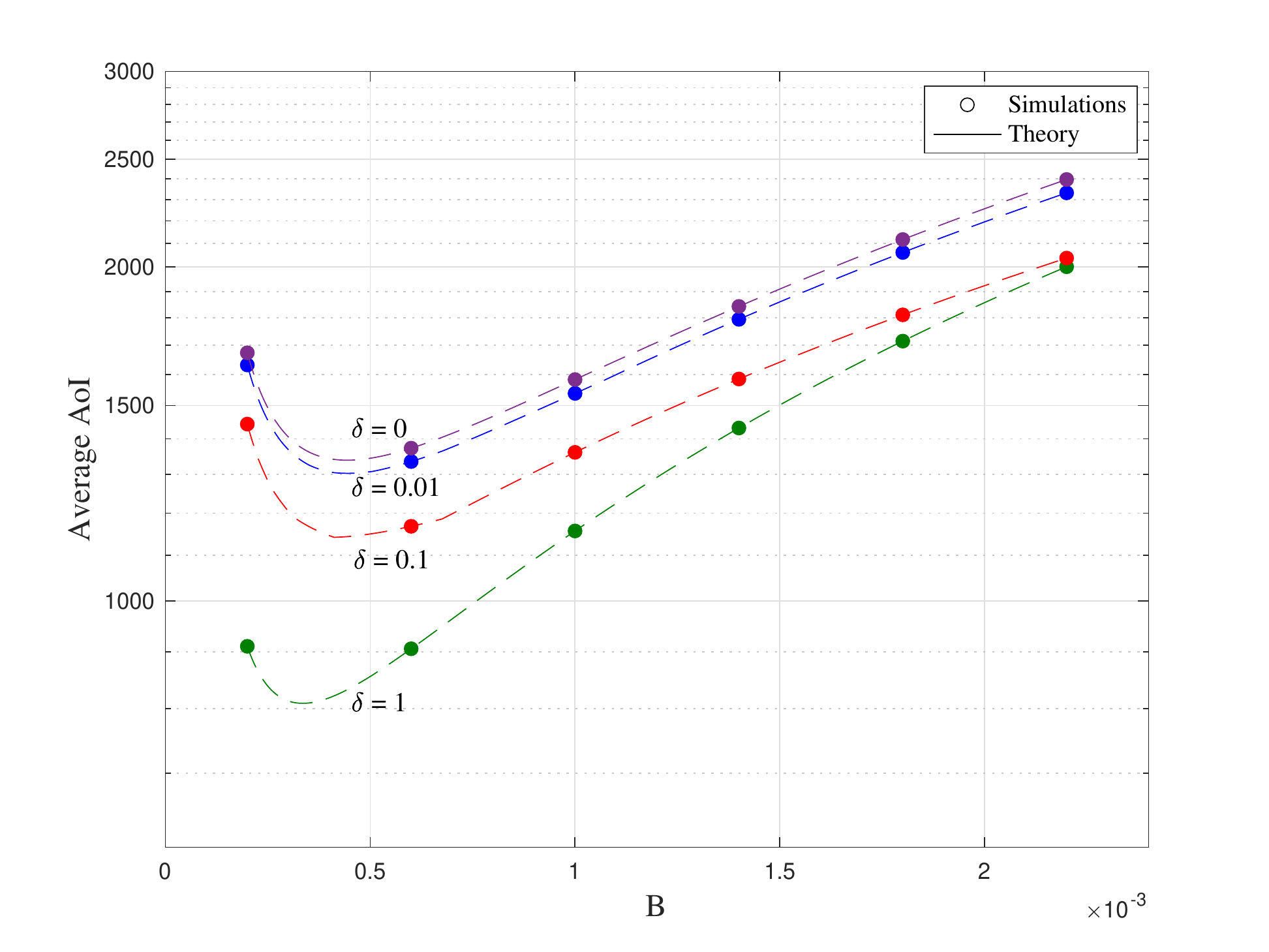}

	\caption{Average AoI with different reliability guarantees for a range of battery capacity.  }\label{fig:delta}

\end{figure}

\begin{table}[H]
	\centering
\begin{tabular}{c|c|ccc}
	\hline
	 &Theory&\multicolumn{3}{|c}{Simulations}  \\
	\hline
	Battery capacity &Reliability (\%)& Statuses sent    & Statuses received & Reliability (\%)    \\
	\hline
	$0.8 \times 10^{-3}$ & $90\%$ & 69181 			& 62220   				& $\approx 89.94\%$\\
	$0.8 \times 10^{-3}$ & $99\%$ & 62933 			& 62314   				& $\approx 99.02\%$\\
	$1 \times 10^{-3}$ & $90\%$ & 58964 			& 53082   				& $\approx 90.02\%$\\
	$1 \times 10^{-3}$ & $99\%$ & 53638 			& 53089  				 & $\approx 98.98\%$\\
	$1.5 \times 10^{-3}$ & $90\%$ & 42844			& 38514						   & $\approx 89.89\%$\\
	$1.5 \times 10^{-3}$ & $99\%$ & 39014 			& 38649  				 & $\approx 99.06\%$\\

	\hline
\end{tabular}
	\caption{Comparison of reliability guarantees of simulations and theoretical results, when $d=20m$ and $P=1W$, for various battery capacities.}
	\label{tab:rel}
\end{table}

\begin{rem}
	Since the theoretical proofs are precise and the expressions are without using any approximations, the relatively small discrepancies between the theoretical results and the outcome of the simulations are due to the small number of statuses used for the simulations. 
\end{rem}

\section{Conclusion}
In this paper, we defined a novel notion, called AoI-reliability trade off. Additionally, we presented two natural schemes in a sensor network, in order to provide a reliability guarantee. Moreover, we presented closed expressions for their average AoI. We also showed that in some sensor networks, by slightly increasing the average AoI, one can achieve a zero-error scheme. Also, we showed that the numerical simulations match our theoretical results. An interesting future direction is to present AoI-reliability trade offs for more involved energy harvesting sensor networks, and in general for any communication system for which reliability is a crucial factor.

\appendices

\section{Omitted claims and proofs}\label{sec:omitted}
\begin{claim}\label{claim:arith}
	For any integer $k\ge 1$ and $\pi\in(0,1)$, we have the following
	\begin{align}
	\sum_{j=1}^{k} (j-1)\cdot \frac{(1-\pi)^{j-1}\pi}{1-(1-\pi)^{k}}
	= \frac{1}{\pi}-\frac{k(1-\pi)^k}{1-(1-\pi)^k}-1.
	\end{align}
\end{claim}
\begin{proof}
	First, note that
	\begin{align}\label{eq:appen1}
	\sum_{j=1}^{k}(j-1)\cdot \frac{(1-\pi)^{j-1}\pi}{1-(1-\pi)^{k}}&= \frac{\pi}{1-(1-\pi)^k}\sum_{j=1}^{k}(j-1)\cdot {(1-\pi)^{j-1}}.
	\end{align}
	Then, we have
	\begin{align}
	\sum_{j=1}^{k}(j-1)\cdot {(1-\pi)^{j-1}}&= \sum_{j=1}^{k}\sum_{i=1}^{j-1}(1-\pi)^{j-1}\nonumber\\
	&=\sum_{i=1}^{k}\sum_{j=i+1}^{k}(1-\pi)^{j-1}\nonumber\\
	&=\sum_{i=1}^{k}(1-\pi)^{i}\sum_{j=0}^{k-i-1}(1-\pi)^j\nonumber\\
	&=\sum_{i=1}^{k}(1-\pi)^{i}\cdot \frac{1-(1-\pi)^{k-i}}{\pi}\nonumber\\
	&=\frac{1-\pi}{\pi}\cdot\frac{1-(1-\pi)^{k}}{\pi}-\frac{k(1-\pi)^{k}}{\pi}\label{eq:appen2}.
	\end{align}
	Now, combining \eqref{eq:appen1} and \eqref{eq:appen2}, we get
	\begin{align}
	\sum_{j=1}^{k}(j-1)\cdot \frac{(1-\pi)^{j-1}\pi}{1-(1-\pi)^{k}}=\frac{1}{\pi}-1-\frac{k(1-\pi)^k}{1-(1-\pi)^k},
	\end{align}
	which concludes the proof.
\end{proof}

\end{document}